\documentclass[11pt]{amsart}
\usepackage{amsfonts}
\usepackage{amsmath}
\usepackage{amssymb}
\usepackage{color}
\usepackage[all]{xy}
\usepackage{graphicx}
\usepackage{xspace}
\usepackage{wrapfig}
\usepackage{soul}
\usepackage{cancel}
\usepackage{ulem}

\newtheorem{theorem}{Theorem}

\newtheorem{corollary}[theorem]{Corollary}

\newtheorem{lemma}[theorem]{Lemma}

\newtheorem{proposition}[theorem]{Proposition}
\newtheorem{remark}[theorem]{Remark}

\newcommand{\un}{\hbox{\bf 1}}
\newcommand{\id}{id}

\begin{document}


\title[${\mathsf{T}}$-ordering and the Magnus expansion]{Time-ordering and a generalized Magnus expansion} 

\date{June 18, 2012}


\author{Michel Bauer}
\address{Institut de Physique Th\'eorique de Saclay 
         			CEA-Saclay, 
			91191 Gif-sur-Yvette, France.}
\email{michel.bauer@cea.fr}
\urladdr{http://ipht.cea.fr/Pisp/michel.bauer/}

\author{Raphael Chetrite}
\address{Laboratoire J.-A.~Dieudonn\'e
         		UMR 6621, CNRS,
         		Parc Valrose,
         		06108 Nice Cedex 02, France.}
\email{Raphael.Chetrite@unice.fr}
\urladdr{http://math.unice.fr/$\sim$rchetrit/}

\author{Kurusch Ebrahimi-Fard}
\address{ICMAT,
		C/Nicol\'as Cabrera, no.~13-15, 28049 Madrid, Spain.
		On leave from UHA, Mulhouse, France}
         \email{kurusch@icmat.es, kurusch.ebrahimi-fard@uha.fr}         
         \urladdr{www.icmat.es/kurusch}

\author{Fr\'ed\'eric Patras}
\address{Laboratoire J.-A.~Dieudonn\'e
         		UMR 6621, CNRS,
         		Parc Valrose,
         		06108 Nice Cedex 02, France.}
\email{patras@math.unice.fr}
\urladdr{www-math.unice.fr/$\sim$patras}


\begin{abstract}
Both the classical time-ordering and the Magnus expansion are well-known in the context of linear initial value problems. Motivated by the noncommutativity between time-ordering and time derivation, and related problems raised recently in statistical physics, we introduce a generalization of the Magnus expansion. Whereas the classical expansion computes the logarithm of the evolution operator of a linear differential equation, our generalization addresses the same problem, including however directly a non-trivial initial condition. As a by-product we recover a variant of the time ordering operation, known as ${\mathsf{T}}^\ast$-ordering. Eventually, placing our results in the general context of Rota--Baxter algebras permits us to present them in a more natural algebraic setting. It encompasses, for example, the case where one considers linear difference equations instead of linear differential equations.
\end{abstract}

\maketitle
\tableofcontents

\section*{Introduction}
\label{sect:intro}

The time-ordered exponential $Y(t) = {\mathsf{T}}\!\exp\bigl(\int_0^t A(s)ds\bigr)$ solves the first order linear differential equation: 
$$
	\dot{Y}(t)= Y(t)A(t),
$$ 
with initial value $Y(0) = \un$. It is complemented by the solution written as a proper exponential of Magnus'  expansion, $Y(t)=\exp\left(\Omega(A)(t)\right)$. Therefore time-ordering and Magnus' expansion are naturally related. Indeed, the function $\Omega(A)(t)$ is simply defined as the logarithm of the time-ordered exponential. It is characterized as the solution of a particular non-linear differential equation:
$$
		\dot{\Omega}(A)(t) = \frac{-ad_{\Omega(A)}}{e^{-ad_{\Omega(A)}} - 1}(A)(t),
$$
and can be calculated recursively, which in turn is accompanied by formidable computational challenges due to intricate algebraic and combinatorial structures.  

Linear initial value problems are abundant in the mathematical sciences. And therefore the time- or  ${\mathsf{T}}$-ordering operation plays an essential role in many fields, for instance in quantum theory and statistical physics. In concrete applications one often needs to calculate time derivations of time-ordered products. However, the noncommutativity of the ${\mathsf{T}}$-ordering operation and time derivation causes obstacles. Our work started out with the aim to understand a simple though surprising paradox caused by the aforementioned noncommutativity. It turns out that physicists introduced a variant of the time-ordering operation, denoted ${\mathsf{T}^\ast}$, which precisely encodes the lack of commutativity between the time-ordering prescription and time derivations. As part of this work we explore the relation between this modified time-ordering operation and the Magnus expansion (Theorem \ref{thm:TRUE}). This lead to a generalization of Magnus' result, allowing to include the initial value into the Magnus expansion (Theorem \ref{thm:RB0}).     

One may consider the time-ordered exponential as the formal solution of the linear integral equation: 
$$
	Y(t)=Y_0+ \int_0^tY(s)A(s)ds
$$
corresponding to the initial value problem. It recently became clear \cite{EM2} that such integral equations, seen as linear fixpoint equations, are just a particular example of a more general class of abstract linear fixpoint equations:
$$
	Y = \un + R(AY) 
$$
defined in associative Rota--Baxter algebras (Theorem \ref{thm:pLMagnus}). We therefore embedded the generalized Magnus expansion into the context of such algebras, by showing in detail a theorem that specializes to the particular case (Theorem \ref{thm:newMagnus}).                 

Let us remark that both time-ordering and Magnus' expansion have triggered much attention and progress in applied mathematics and physics. See for instance \cite{BCOR09,IMKNZ,MielPleb,Strichartz,Wilcox}. Recently, however, Magnus' work has also been explored successfully from a more algebraic-combinatorial perspective using operads, rooted trees, Hopf algebras, pre-Lie algebras, and the theory of noncommutative symmetric functions \cite{CP,EM1,EM2,gelfand,Murua}.
 
\smallskip 

Finally, as a word of warning, we should mention that we deliberately ignore any convergence issues. This work addresses mainly formal algebraic and related combinatorial aspects.

\medskip 

The paper is organized as follows. In the first section we introduce the objects we are going to consider in this work, i.e.~the time-ordering map and the Magnus expansion. In section \ref{sect:TordTime} we introduce a variant of the time-ordering operation, which reflects the non-commutativity of time-ordering and time derivation. The connection to the Magnus expansion is detailed. Section \ref{sect:LinRecRB} recalls briefly the notion of Rota--Baxter algebras, and we then embed the foregoing results into this context. We also show a mild generalization of Atkinson's factorization theorem for Rota--Baxter algebras. The notion of $R$-center is introduced, which plays an important role. Finally, in section \ref{ssect:finitediff} we present finite differences as a simple example.


\section{Time-ordering and the Magnus expansion}
\label{sect:TordMagnus}


\subsection{The ${\mathsf{T}}$-ordering map}
\label{ssect:Tord}

Linear initial value problems (IVP) play a major role in many parts of applied mathematics and theoretical physics. They consist of an ordinary linear differential equation for an unknown function $Y(t)$:
\begin{align}
	\frac{d}{dt} Y(t)  =Y(t)A(t)
	 \label{eq:IVP}
\end{align}
together with a specified value of this function at a given point in time, e.g.~say for $t=0$, $Y(0)  = Y_0$, which we call the initial value. The functions $Y(t)$ and $A(t)$ being for instance matrix valued and differentiable. 

Further simplifying the problem, e.g.~by restricting to scalar valued functions, the solution of (\ref{eq:IVP}) is given straightforwardly  in terms of the exponential map:
\begin{equation}
\label{eq:simple}
	Y(t)=Y_0\exp\left(\int_0^t A(s) ds\right).
\end{equation}
Returning to the more challenging case, when noncommutativity prevails, we see that the result is far from simple. Recall that the formal solution of (\ref{eq:IVP}) is given by the series of iterated integrals:       
\begin{equation}
\label{eq:Torder}
	 Y(t) = Y_0\left( \un + \int_0^t A(s_1) ds_1 \
	   + \int_0^t \int_0^{s_1} A(s_2) ds_2\ A(s_1)ds_1\
	   + \cdots \right),
\end{equation}
where $\un$ would be the identity matrix in the context of matrix valued differentiable functions. The right hand side, known as Dyson--Chen or ${\mathsf{T}}$(ime)-ordered series, follows forthrightly from the solution of the integral equation: 
\begin{eqnarray}
\label{eq:recursion}
	Y(t)=Y_0+ \int_0^tY(s)A(s)ds
\end{eqnarray}
corresponding to (\ref{eq:IVP}). Note that in the commutative case series (\ref{eq:Torder}) simply follows from expanding the exponential map (\ref{eq:simple}) and making use of the integration by parts rule:
\begin{eqnarray}
\label{eq:RB0}
	\int_0^tf(s)ds\ \int_0^tg(u)du = \int_0^t\int_0^sf(s)g(u)du\, ds + \int_0^t\int_0^uf(s) g(u)ds\,du. 
\end{eqnarray}

Let us recall the notion of time- or ${\mathsf{T}}$-ordering. Time ordering is defined in terms of the ${\mathsf{T}}$-map:
\begin{equation*}
	{\mathsf{T}}[U(s_1)V(s_2)] = \begin{cases} 	U(s_1)V(s_2),& \quad s_1<s_2\\
					 				V(s_2)U(s_1),&\quad s_2<s_1.
						    \end{cases}
\end{equation*}
Written more concisely it involves Heaviside step functions \cite{Fried}:
$$
	{\mathsf{T}}[U(s_1)V(s_2)]:=\Theta(s_2-s_1)U(s_1)V(s_2) + \Theta(s_1-s_2)V(s_2)U(s_1),
$$
with a generalization to higher products of operators:
$$
	{\mathsf{T}}[U_1(s_1) \cdots U_n(s_n)]:=\sum_{\sigma \in S_n} \Theta^\sigma_n 
								U_{\sigma(1)}(s_{\sigma(1)}) \cdots U_{\sigma(n)}(s_{\sigma(n)}),
$$
where $\Theta^\sigma_n:= \prod_{i=1}^n  \Theta(s_{\sigma(1)}-s_{\sigma(2)}) \cdots \Theta(s_{\sigma(n-1)}-s_{\sigma(n)})$. Observe that for any permutation $\sigma \in S_n$, ${\mathsf{T}}[U_1(s_1) \cdots U_n(s_n)]={\mathsf{T}}[U_{\sigma(1)}(s_{\sigma(1)}) \cdots U_{\sigma(n)}(s_{\sigma(n)})]$. Following the convention that the ${\mathsf{T}}$-map operates before integration, the solution of the IVP (\ref{eq:IVP}) writes:
$$
	Y(t)=Y_0{\mathsf{T}}\!\exp\left(\int_0^t A(s) ds\right)
	     :=Y_0\sum_{n\ge 0} \frac{1}{n!}\int_0^t \cdots \int_0^t\  {\mathsf{T}}[A(s_1)\cdots A(s_n)]ds_1\cdots ds_n.
$$   
Each $n$-fold ${\mathsf{T}}$-product inside the integral produces $n!$ strictly iterated integrals, and therefore the series coincides with (\ref{eq:Torder}).   

As a remark we would like to mention that our work started out with the aim to understand a simple though surprising paradox caused by the noncommutativity of the ${\mathsf{T}}$-ordering operation and time derivation. See \cite{BGL} for an interesting example where this plays a decisive role, and \cite{Talkner} for a case where this paradox caused some confusion. We assume nonnegative $t$, and $A_t:=A(t)$ may be a finite-dimensional matrix valued function, or more generally a linear operator valued map. The following identity:
 \begin{equation*}
	\exp\left(A_{t}-A_{0}\right)=\exp\left(\int_{0}^{t}du\frac{dA_{u}}{du}\right)
\end{equation*}
may be considered as evident. However, its ${\mathsf{T}}$-ordering cousin, which would result from a naive action of the ${\mathsf{T}}$-ordering operation on both sides of the foregoing equation: 
\begin{equation}
\label{eq:paradox}
	{\mathsf{T}}\!\exp\left(A_{t}-A_{0}\right)={\mathsf{T}}\!\exp\left(\int_{0}^{t}du\frac{dA_{u}}{du}\right)
\end{equation}
does not hold. Indeed, let us look at the second order term of each side separately:
$$
	{\mathsf{T}}\!\bigl[(A_{t}-A_{0})^2\bigr]=A_0A_0 - 2 A_0A_t + A_tA_t,
$$
whereas:
\allowdisplaybreaks{
\begin{eqnarray*}
	{\mathsf{T}}\!\left[(\int_{0}^{t}du\frac{dA_{u}}{du})^{2}\right] 
	& = & \int_{0}^{t}du\int_{u}^{t}dv\frac{dA_{u}}{du}\frac{dA_{v}}{dv}
			+\int_{0}^{t}du\int_{0}^{u}dv\frac{dA_{v}}{dv}\frac{dA_{u}}{du}\\
 	& = & 2\int_{0}^{t}du\int_{u}^{t}dv\frac{dA_{u}}{du}\frac{dA_{v}}{dv}
	   =	2\int_{0}^{t}du\frac{dA_{u}}{du}(A_{t}-A_{u})\\
 	& = & 2(A_{t}-A_{0})A_{t}-2\int_{0}^{t}du\frac{dA_{u}}{du}A_{u}.
 \end{eqnarray*}}
Hence, we have verified that:
\begin{equation*}
	{\mathsf{T}}\!\bigl[(A_{t}-A_{0})^{2}\bigr] 
	\neq {\mathsf{T}}\!\left[\bigl(\int_{0}^{t}du\frac{dA_{u}}{du}\bigr)^{2}\right].
\end{equation*}

In fact, it is clear that (\ref{eq:paradox}) does not hold due to the lack of commutativity of the ${\mathsf{T}}$-ordering with time derivation, which, on the other hand, is caused by the Heaviside step functions involved in the definition of the ${\mathsf{T}}$-ordering map. However, observe that the difference reads:
\allowdisplaybreaks{
\begin{eqnarray*}
	\lefteqn{{\mathsf{T}}\!\bigl[\left(A_{t}-A_{0}\right)^{2}\bigr]
		-{\mathsf{T}}\!\left[\bigl(\int_{0}^{t}du\frac{dA_{u}}{du}\bigr)^{2}\right] 
		= -A_{t}A_{t}+A_{0}A_{0}+2\int_{0}^{t}du\frac{dA_{u}}{du}A_{u}}\nonumber \\
 		& = & -A_{t}A_{t}+A_{0}A_{0}+\int_{0}^{t}du\left(\frac{dA_{u}}{du}A_{u}
				+A_{u}\frac{dA_{u}}{du}\right)+\int_{0}^{t}du\left(\frac{dA_{u}}{du}A_{u}
				-A_{u}\frac{dA_{u}}{du}\right)\nonumber \\
 		& = & -A_{t}A_{t}+A_{0}A_{0}+A_{t}A_{t}-A_{0}A_{0}
			+\int_{0}^{t}du\left(\frac{dA_{u}}{du}A_{u}-A_{u}\frac{dA_{u}}{du}\right)\nonumber \\
 		& = & \int_{0}^{t}du\left[\frac{dA_{u}}{du},A_{u}\right].
 \end{eqnarray*}}
In the sequel we will further remark on this observation, and explore in more detail the object, denoted $\Phi\bigl(\frac{dA_{u}}{du}\bigr)$, such that the following equality holds:
\begin{equation}\label{defPhi}
	{\mathsf{T}}\!\exp\left(A_{t}-A_{0}\right)={\mathsf{T}}\!\exp\left(\int_{0}^{t}du\ \Phi\bigl(\frac{dA_{u}}{du}\bigr)\right).
\end{equation}


\subsection{The Magnus expansion}
\label{ssect:Magnus}

It is clear that the proper exponential solution of (\ref{eq:IVP}) changes drastically in the light of the noncommutative character of the problem. Wilhelm Magnus proposed in his seminal 1954 paper \cite{Magnus} a particular differential equation for the logarithm of the series (\ref{eq:Torder}) of iterated integrals (with initial value $Y_0=\un$), which we denote by the function $\Omega(A)(t)$:
$$
	\dot{\Omega}(A)(t) =A(t) +  \sum_{n>0} (-1)^n\frac{B_n}{n!} ad_{\int_0^t \dot{\Omega}(A)(s)ds}^{(n)}(A(s)) 
					= \frac{-ad_{\Omega(A)}}{e^{-ad_{\Omega(A)}} - 1}(A)(t),
$$ 
with $\Omega(A)(0)=0$. The $B_n$ are the Bernoulli numbers $B_0=1,\ B_1=-\frac 12,\ B_2=\frac 16,\ldots \ {\rm{and}}\;\; B_{2k+1}=0 \hbox{ for }k\ge 1$, such that the solution of the IVP (\ref{eq:IVP}) (with $Y_0=\un$) is given by: 
\begin{equation}
	Y(t)=\exp\left(\int_0^t \dot{\Omega}(A)(s)ds\right)
	\label{magclas}.
\end{equation}
Here, as usual, the $n$-fold iterated Lie bracket is denoted by $ad^{(n)}_a(b):=[a,[a,\cdots [a,b]]\cdots]$. 
Recall that:
$$
	\frac{-z}{\exp(-z)-1}=\sum_{n \geq 0} (-1)^n\frac{B_n}{n!} z^n 
	\quad\ {\rm{and}}\ \quad
	\frac{\exp(-z)-1}{-z}=\int^1_0 \exp(-sz)ds.
$$ 
The reader is encouraged to consult \cite{BCOR09} for an authoritative review on the Magnus expansion and its ramifications. References \cite{MielPleb,Strichartz,Wilcox} are classical sources regarding Magnus' work. Let us write down the first few terms of what is called Magnus' series, $\Omega(\lambda A)(s)=\sum_{n \ge 1} \lambda^n \Omega_{n}(A)(s)$, following from Picard iteration to solve the above recursion:
\allowdisplaybreaks{
\begin{eqnarray*}
	\frac{d}{dt}{\Omega}(\lambda A)(t)
			&=&\lambda A(t) 
				+ \lambda^2 \frac{1}{2}\left[\int_0^t\!\!\!\!A(s)ds ,A(t)\right] \\
			&&\quad +\lambda^3 \frac{1}{4}\left[\int_0^t\!\Bigl[\int_0^s\!\!\!\!A(u)du ,A(s)\Bigr]ds,A(t)\right]
		+\lambda^3 \frac{1}{12}\left[\int_0^t\!\!\!\!A(s)ds,\Bigl[\int_0^t\!\!\!\!A(u)du ,A(t)\Bigr]\right] +\cdots,
\end{eqnarray*}}
where we introduced a parameter $\lambda$ for later convenience. In the light of the IVP (\ref{eq:IVP}) and its two solutions, i.e.~the proper exponential solution using Magnus' expansion, and the ${\mathsf{T}}$-ordered exponential, we find the identity (for $Y_0=\un$):
\begin{equation}
\label{magnus}
 	{\mathsf{T}}\!\exp\left(\int_0^t A(s) ds\right) = \exp\left(\int_0^t \dot{\Omega}(A)(s)ds\right). 
\end{equation}
The terms beyond order one in $\Omega(\lambda A)(s)= \lambda\int_0^t A(s) ds + \sum_{n \ge 2} \lambda^n \Omega_{n}(A)(s)$ consist of iterated Lie brackets and integrals. For each $n>1$ they solely serve to rewrite higher powers $(\int_0^t A(s) ds)^n$ as strictly iterated integral of order $n$, modulo a factor $n!$. For instance, verify that:
$$
	\frac{1}{2}\left(\int_0^t A(s) ds\right)^2 +  \frac{1}{2}\int_0^t \left[\int_0^s A(u)du ,A(s)\right]ds 
	= \int_0^t  \int_0^s A(u) du\ A(s)ds
$$   

We remark that recently, Magnus' work has been successfully explored from a more algebraic-combinatorial perspective using the theory of operads, rooted trees, pre-Lie algebras, Rota--Baxter algebras, and the theory of noncommutative symmetric functions \cite{CP,EM1,EM2,gelfand,IMKNZ,Murua}.


\section{${\mathsf{T}}$-ordering, ${\mathsf{T}}^\ast$-ordering and time derivations}
\label{sect:TordTime}

\subsection{${\mathsf{T}}$-ordering and time derivations}
\label{ssect:TordTime}

Let us return to the erroneous identity (\ref{eq:paradox}). Recall that to simplify the presentation, we sometimes write $A_t$ for $A(t)$. 

\begin{theorem} \label{thm:TRUE}
The solution $\Phi$ to the equation:
\begin{equation*}
	{\mathsf{T}}\!\exp\left(A_{t}-A_{0}\right)=
	{\mathsf{T}}\!\exp\left(\int_{0}^{t}du\ \Phi\bigl(\frac{dA_{u}}{du}\bigr)\right)
\end{equation*}
is explicitly given by:
\begin{equation}
	\Phi\bigl(\frac{dA_{u}}{du}\bigr)
	:=\sum_{n=1}^{\infty}\frac{1}{n!}\sum_{m=0}^{n-1}\left(-1\right)^{m}C_{n-1}^{m}
		\left(A_{u}\right)^{m}\frac{dA_{u}}{du}\left(A_{u}\right)^{n-1-m}
	 =\frac{e^{-ad_{A_u}}-1}{-ad_{A_u}}\left(\frac{dA_u}{du}\right)
\label{eq:goodT}
\end{equation}
\end{theorem}

Observe that if $A_u$ commutes with its derivation, $\left[\frac{dA_{u}}{du},A_{u}\right]=0$, then, as expected equation (\ref{eq:goodT}) reduces to $\Phi\bigl(\frac{dA_{u}}{du}\bigr)=\frac{dA_{u}}{du}$.

In fact, equation (\ref{eq:goodT}) follows from exploring briefly what went wrong in the first place with identity (\ref{eq:paradox}). Indeed, recall that $X_t = {\mathsf{T}}\!\exp\left(\int_{0}^{t}duB_{u}\right)$ for a family of operator valued maps $B_t$ is the solution to the simple IVP $\frac{d}{dt}X_{t} =  X_{t}B_{t}$, $X_{0} =  \un.$ Hence:
\begin{equation*}
	\frac{d}{dt}{\mathsf{T}}\!\exp\left(\int_{0}^{t}du\frac{dA_{u}}{du}\right)
	={\mathsf{T}}\!\exp\left(\int_{0}^{t}du\frac{dA_{u}}{du}\right)\frac{dA_{t}}{dt}.
\end{equation*}
On the other hand it is clear from its definition, that:
\begin{equation*}
	{\mathsf{T}}\!\exp\left(A_{t}-A_{0}\right)=\exp\left(-A_{0}\right)\exp\left(A_{t}\right).
\end{equation*}
Therefore:
\allowdisplaybreaks{
\begin{eqnarray}
	\frac{d}{dt}{\mathsf{T}}\!\exp\left(A_{t}-A_{0}\right) 
	& = & \exp\left(-A_{0}\right)\frac{d}{dt}\exp\left(A_{t}\right)
	    = \exp\left(-A_{0}\right)\exp\left(A_{t}\right)\exp\left(-A_{t}\right)\frac{d}{dt}\exp\left(A_{t}\right)\label{eq:forint}\\
 	& = & \left({\mathsf{T}}\!\exp\left(A_{t}-A_{0}\right)\right)\exp\left(-A_{t}\right)\frac{d}{dt}\exp\left(A_{t}\right), \nonumber
\end{eqnarray}}
which again shows that (\ref{eq:paradox}) can not be correct. 

However, the last equality may be used to derive (\ref{eq:goodT}). Indeed (\ref{eq:forint}) implies:
\begin{equation*}
	{\mathsf{T}}\!\exp\left(A_{t}-A_{0}\right)
	={\mathsf{T}}\!\exp\left(\int_{0}^{t}du\exp\left(-A_{u}\right)\frac{d}{du}\exp\left(A_{u}\right)\right).
\end{equation*}
Using Duhamel's formula:
\begin{equation}
	\exp\left(-A_{u}\right)\frac{d}{du}\exp\left(A_{u}\right)
			=\int_{0}^{1}d\nu\exp\left(-\nu A_{u}\right)\frac{dA_{u}}{du}\exp\left(\nu A_{u}\right),
\label{eq:duh}
\end{equation}
and recalling that: 
$$
	\int_{0}^{1}d\nu\exp(-\nu a)b\exp(\nu a)=
	\sum_{i \geq 1} \frac{\left(-1\right)^{i-1}}{i!}ad_a^{(i-1)}(b)=\frac{e^{-ad_a}-1}{-ad_a}(b),
$$
yields:
\begin{equation}
	\exp\left(-A_{0}\right)\exp\left(A_{t}\right)=\mathsf{T}\exp\left(A_{t}-A_{0}\right)
	 =  {\mathsf{T}}\!\exp\left(\int_{0}^{t}du\frac{e^{-ad_{A_u}}-1}{-ad_{A_u}}\Bigl(\frac{dA_u}{du}\Bigr)\right).
\label{fred}
\end{equation}
This implies (\ref{eq:goodT}) because $ad_{a}^{(m)}(b)=\sum_{n=0}^{m-n}\left(-1\right)^{n}C_{m}^{n}a^{n}ba^{m-n}$. Conversely, if we set $B_t:=\Phi\bigl(\frac{dA_{u}}{du}\bigr)$ and $X_t:=\exp(A_t)$, we arrive at the following first generalization of the classical result of Magnus (\ref{magclas}).

\begin{theorem} \label{thm:RB0}
The IVP: $ \dot{X}_{t} =  X_{t}B_{t}$, $X_{0} =  \exp(\alpha)$ is solved by:
\begin{equation*}
	X_t=\exp\left(\alpha + \int_0^t \dot{\Omega}_\alpha(B)(s)ds\right),
\end{equation*}
with:
\begin{equation*}
	\dot{\Omega}_\alpha(B)(u):= \frac{-ad_{\alpha + \int_0^u \dot{\Omega}_\alpha(B)(s)ds}}
							   {e^{-ad_{\alpha+ \int_0^u \dot{\Omega}_\alpha(B)(s)ds}}-1}(B(u)).
\end{equation*}\\
\end{theorem}


\subsection{${\mathsf{T}}^\ast$-ordering and the Magnus formula}

Recall that the modified time- or ${\mathsf{T}}^\ast$-ordering is defined precisely so as to encode appropriately the lack of commutativity between the ${\mathsf{T}}$-ordering and time derivations. In \cite{BGL} the interested reader can find motivations for the introduction of this particular variant of the time ordering operator coming from statistical physics, as well as further bibliographical references on the subject.

For $I(t):=\frac{dQ(t)}{dt}$, we define:
\begin{equation*}
	{\mathsf{T}}^\ast[I(t_1) \cdots I(t_n)]
	:=\frac{\partial}{\partial t_n} \cdots \frac{\partial}{\partial t_1}{\mathsf{T}}[Q(t_1) \cdots Q(t_n)].
\end{equation*}
For any permutation $\sigma \in S_n$: 
\begin{equation*}
	{\mathsf{T}}^\ast[I(t_1) \cdots I(t_n)]={\mathsf{T}}^\ast[I(t_{\sigma(1)}) \cdots I(t_{\sigma(n)})],
\end{equation*} 
or, equivalently, we have an invariance under the right action of the symmetric group, $S_n$, in $n$ elements, ${\mathsf{T}}^\ast={\mathsf{T}}^\ast\circ{\sigma }$. In particular, because of this property, ${\mathsf{T}}^\ast$ is entirely characterized by its action on time-ordered products $I(t_1) \cdots I(t_n)$, $t_1\leq ...\leq t_n$, or, equivalently on products $I(t_1)^{m_1} \cdots I(t_n)^{m_n}$, $t_1< \cdots < t_n$.

Therefore we get, by setting $I(t)=\dot{A}(t)$ and $Q(t)=A(t)-A(0)$:
\begin{eqnarray*}
	{\mathsf{T}}^\ast\!\exp\left(\int_{0}^{t}du\dot{A}_{u}\right)
	&=&\sum_n\frac{1}{n!}\int\limits_0^t \cdots \int\limits_0^t \prod_{i=1}^ndt_i
	\frac{\partial}{\partial t_n}\cdots \frac{\partial}{\partial t_1}{\mathsf{T}}\bigl[(A(t_1)-A(0)) \cdots (A(t_n)-A(0))\bigr]\\
	&=& \sum_n\frac{1}{n!}{\mathsf{T}}\bigl[(A(t)-A(0)) \cdots (A(t)-A(0))\bigr]\\
	&=&\sum_n\sum_{k\leq n}\frac{1}{n!}{n\choose{k}}(-1)^kA(0)^kA(t)^{n-k}=\exp(-A(0))\exp(A(t)).
\end{eqnarray*}
Equation (\ref{fred}) yields:
\begin{equation}
\label{tstar}
	{\mathsf{T}}^\ast\!\exp\left(\int_{0}^{t}du\dot{A}_{u}\right)
	={\mathsf{T}}\!\exp\left(\int_{0}^{t}du\ \frac{e^{-ad_{A_u}}-1}{-ad_{A_u}}\bigl(\dot{A_u}\bigr)\right).
\end{equation}

Let us look at the degree two terms in $A$ and $\dot{A}$ on both sides of this equality. We find:
$$
	\frac{1}{2!} \int_0^t\int_0^t T^*[\dot{A}_{u_1}\dot{A}_{u_2}] du_1du_2 
	= \frac{1}{2!} \int_0^t\int_0^t T[\dot{A}_{u_1}\dot{A}_{u_2}] du_1du_2 + \frac{1}{2!}\int_0^t ad_{A_{u_1}}(\dot{A}_{u_1})du_1.
$$ 
Simply rewriting the last term on the right hand side:
$$
	 \frac{1}{2!}\int_0^t ad_{A_{u_1}}(\dot{A}_{u_1})du_1 = \frac{1}{2!} \int_0^t\int_0^t ad_{A_{u_1}}(\dot{A}_{u_2}) \delta(u_1-u_2)du_1du_2 
$$
we get finally, for $u_1$ different from $u_2$:
$$
	T^*[\dot{A}_{u_1}\dot{A}_{u_2}] 
	= T[\dot{A}_{u_1}\dot{A}_{u_2}] .
$$
On the diagonal, i.e.~for $u_1=u_2$ we obtain:
$$
	T^*[\dot{A}_{u_1}\dot{A}_{u_2}] 
	= ad_{A_{u_1}}(\dot{A}_{u_2}) \delta(u_1-u_2) .
$$
In general, by identifying homogeneous components on each side of the foregoing equation, and taking into account the permutational invariance of ${\mathsf{T}}^\ast$ (so that e.g.~ $T^*[\dot{A}_{u_1}\dot{A}_{u_2}\dot{A}_{u_1}]=T^*[\dot{A}_{u_1}^2\dot{A}_{u_2}]$), we finally arrive at:
\begin{equation*}
	\frac{1}{n!}{n\choose n_1,\ldots,n_k}{\mathsf{T}}^\ast \bigl[\dot A_{u_1} \cdots \dot A_{u_n}\bigr]=
	{\mathsf{T}}\left[(-1)^{n_1}\frac{ad_{A_{c_1}}^{(n_1-1)}}{n_1!}
	(\dot A_{c_1}) \cdots (-1)^{n_k}\frac{ad_{A_{c_k}}^{(n_k-1)}}{n_k!}(\dot A_{c_k})\right]\delta_{1}\cdots\delta_{k},
\end{equation*}
for arbitrary $0\leq c_1< \cdots <c_k\leq t$, and $n_1,\ldots,n_k$, such that $n=n_1+ \cdots +n_k$, $u_1= \cdots =u_{n_1}=c_1, \ldots ,u_{n_1+ \cdots +n_{k-1}+1}= \cdots =u_{n_1+ \cdots +n_{k-1}+n_k}=c_k$. Recall that the multinomial coefficient on the left hand side stands for $n!(n_1!\cdots n_k!)^{-1}$. Here we wrote $\delta_i$ for the product of delta functions $\delta(u_{n_1+ \cdots +n_{i-1}+1}- u_{n_1+ \cdots +n_{i-1}+2}) \cdots \delta(u_{n_1+ \cdots +n_{i-1}+n_i-1} -u_{n_1+ \cdots +n_{i-1}+n_i})$ encoding the change from a $n$-dimensional integral to a $k$-dimensional integral from the left to the right hand side of equation (\ref{tstar}). Note that we have recovered \cite[eq.~2.11]{BGL} (up to signs, due to different conventions regarding time-orderings):
\begin{corollary}
We have, with our previous notation:
\begin{equation}
\label{tstareq}
	{\mathsf{T}}^\ast[\dot A_{u_1}\cdots\dot A_{u_n}]
	=(-1)^{n} {\mathsf{T}}[{ad_{A_{c_1}}^{(n_1-1)}}(\dot A_{c_1}) \cdots ad_{A_{c_k}}^{(n_k-1)}(\dot A_{c_k})]
	\delta_{1} \cdots \delta_{k}  
\end{equation}\\
\end{corollary}


\section{Linear recursions in Rota--Baxter algebras}
\label{sect:LinRecRB}

We take now a broader perspective on linear IVP by replacing (\ref{eq:recursion}) in terms of general linear operator fixpoint equations. A natural setting to do so is provided by associative Rota--Baxter $k$-algebra, which we briefly recall in this section. For more details we refer the reader to \cite{Atkinson,Baxter,EGP,EMP,RS} and references therein. See also \cite{Aczel,AD} for interesting aspects. We assume the underlying base field $k$ to be of characteristic zero, e.g.~$k = \mathbb{R}$ or $\mathbb{C}$. 

Motivated by Frank Spitzer's seminal work \cite{Spitzer} in probability theory, the mathematician G.~Baxter~\cite{Baxter}, soon after followed by J.~F.~C.~Kingman, W.~Vogel, F.~Atkinson, P.~Cartier, Gian-Carlo Rota and others \cite{Atkinson,Cartier,Kingman,Vogel,Rota1,Rota2,RS}, explored a more general approach to the above linear IVP (\ref{eq:IVP}) and the corresponding linear integral equation (\ref{eq:recursion}) -- in the case of $Y_0=\un$ --, by suggesting to work with the linear fixpoint equation:
\begin{equation}
\label{eq:RBrecursion}
    Y = \un + \lambda R(Ya),
\end{equation}
in a commutative unital $k$-algebra $\mathcal{A}$, $a \in \mathcal{A}$, respectively the filtered and complete algebra $\mathcal{A}[[\lambda]]$. The linear map $R$ on $\mathcal{A}$ is supposed to satisfy the Rota--Baxter relation of scalar weight $\theta \in k$:
\begin{equation}
\label{RBR}
    R(x)R(y) = R( R(x)y+xR(y)) + \theta R(xy).
\end{equation}
We call $\mathcal{A}$ a commutative Rota--Baxter $k$-algebra of weight $\theta \in k$ with unit $\un$. Observe that the map $\tilde{R} := -\theta \id - R$ is also a Rota--Baxter map of weight $\theta$. Moreover, from identity (\ref{RBR}) it follows that the images of $R$ and $\tilde{R}$ are subalgebras in $\mathcal{A}$.

Baxter gave a solution for (\ref{eq:RBrecursion}) by proving Spitzer's classical identity in a commutative unital Rota--Baxter $k$-algebra of weight $\theta \in k$, that is, he showed that the formal solution of (\ref{eq:RBrecursion}) is given by:
\begin{equation}
\label{spitzer}
	Y = \un + \sum\limits_{n=1}^\infty
  		\underbrace{R\bigl( R( \cdots R(R}_{n\mbox{\rm -} {\rm times}}(a) a )a \cdots )a \bigr)
	    = \exp\left(-R\Bigl(\frac{\log(1- \theta a)}{\theta}\Bigr) \right).
\end{equation}

\begin{remark}\label{rmk:BaAtVo}
{\rm{Note that besides (\ref{eq:RBrecursion}), Baxter, Atkinson and Vogel considered more general linear fixpoint equations and their solutions in a commutative Rota--Baxter $k$-algebra $\mathcal{A}$. In \cite{Atkinson,Baxter} the equations $E = a + \lambda R(bE)$, $F = R(a) + \lambda R(bF)$, with analog equations for $\tilde{R}$, and  $G= a + \lambda  R(bG) +  \lambda \tilde{R}(cG)$ are solved for $a,b,c \in \mathcal{A}$. In \cite{Vogel} Vogel considered the equation $H= a + \lambda  bR(H) + c \lambda \tilde{R}(H)$, and gave a solution for $a,b,c \in \mathcal{A}$. The generalization of these equations and their solutions to non-commutative Rota--Baxter algebra was given in \cite{EM2,EM3}.}}
\end{remark}

Let us return to (\ref{RBR}), which one may think of as a generalized integration by parts identity. Indeed, the Riemann integral, which satisfies the usual integration by parts rule (\ref{eq:RB0}), corresponds to a Rota--Baxter operator of weight zero. Observe that in the limit $\theta \to 0$ we obtain:
$$
    \frac{1}{\theta} \log(1 - \theta a)
     = -\sum\limits_{n>0} \theta^{n-1} \frac{a^n}{n}\ \xrightarrow{\quad\theta \to 0 \quad}\ -a.
$$
The right hand side of the last equality in (\ref{spitzer}) reduces to the simple exponential solution $Y=\exp(R(a))$. In the context of the Riemann integral, this corresponds to the classical case (\ref{eq:simple}) of the IVP (\ref{eq:IVP}), with $Y_0=1$.

The extra term on the right hand side of (\ref{RBR}) becomes necessary, for instance, when we replace the Riemann integral by a Riemann-type summation operator: 
\begin{equation}
\label{sum}
    S(f)(x):= \sum_{n>0} f(x+n)
\end{equation}
on a suitable class of functions. Indeed, one verifies that the map $S$ satisfies the weight $\theta=1$ Rota--Baxter relation:
\begin{equation*}
     S(f) S(g) = S\big(S(f) \: g\big) + S\big(f \:S(g)\big) + S\big(fg\big).
\end{equation*}
More generally for finite Riemann sums: 
\begin{eqnarray}
\label{Riemsum}
    R_\theta(f)(x) := \sum_{n = 0}^{[x/\theta]-1} \theta f(n\theta).
\end{eqnarray}
we find that $R_\theta$ satisfies the weight $\theta$ Rota--Baxter relation:
$$ 
	R_\theta(f)R_\theta(g)(x)=
	R_\theta\bigl(R_\theta(f)g\bigr)(x) + R_\theta\bigl(fR_\theta(g)\bigr)(x) + \theta R_\theta(fg)(x).
$$
Beside fluctuation theory in probability the Rota--Baxter relation recently played a crucial role in Connes--Kreimer's Hopf algebraic approach to perturbative renormalization \cite{CK}. Here, Rota--Baxter algebras enter through projectors, which satisfy (\ref{RBR}) for $\theta=-1$. A paradigm is provided by the algebra $\mathcal{A}=\mathbb{C}[\varepsilon^{-1},\varepsilon]]=\varepsilon^{-1}\mathbb{C}[\varepsilon^{-1}]\oplus \mathbb{C}[[\varepsilon]] $ of Laurent series $ \sum_{n=-k}^\infty a_n \varepsilon^{n}$ (with finite pole part). The projector, called minimal subtraction scheme map, used in renormalization is defined by keeping the pole part:
\begin{equation*}
    R_{ms}\Big(\sum_{n=-k}^\infty a_n \varepsilon^{n}\Big) 
    =  R_{ms}\Big(\sum_{n=-k}^{-1} a_n \varepsilon^{n} + \sum_{n=0}^\infty a_n \varepsilon^{n}\Big)  
    := \sum_{n=-k}^{-1} a_n \varepsilon^{n}.
\end{equation*}
It is a Rota--Baxter map of weight $\theta=-1$. In general, assume the $k$-algebra $\mathcal{A}$ decomposes, $\mathcal{A}  = \mathcal{A} _1 \oplus \mathcal{A} _2$, and let $R: \mathcal{A}  \to \mathcal{A} $ be defined by $R(a_1,a_2)=a_1$, then $R^2=R$. Hence: 
\begin{equation*}
    R(a)b + aR(b) - ab = R(a)R(b) - (\id-R)(a)(\id-R)(b)
\end{equation*}
such that applying $R$ yields $R\big( R(a)b + aR(b) - ab \big) = R(a)R(b)$. We refer the reader to \cite{EMP} for more details including examples and applications. 

In \cite{EM2, EM3} the Spitzer identity has been generalized to non-commutative Rota--Baxter algebras:

\begin{theorem} \cite{EM2} \label{thm:pLMagnus} 
Let $(\mathcal{A},R)$ be a unital Rota--Baxter $k$-algebra of weight $\theta \in k$. Let $\Omega':=\Omega'(\lambda a)$, $a \in \mathcal{A}$, be the element of $\lambda A[[\lambda]]$ such that:
$$
	Y=\exp\bigl(R(\Omega')\bigr),
$$ 
where $Y=Y(a)$ is the solution of the linear fixpoint equation $Y = \un + \lambda R(Ya)$. This element obeys the following recursive equation:
\begin{equation}
\label{eq:pLMagnus}
    \Omega'(\lambda a) = \frac{-ad_{R(\Omega')} + r_{\theta \Omega'}}{e^{-ad_{R(\Omega')} + r_{\theta \Omega'}}-1}(\lambda a)
            =\sum\limits_{m\ge 0} (-1)^m \frac{B_m}{m!}\ \widetilde{ad}^{(m)}_{\Omega'}(\lambda a)
\end{equation}
with $B_l$ the Bernoulli numbers, $r_{\theta a}(b):= \theta ba$, and $\widetilde{ad}_a(b):=ad_{R(a)}(b) - r_{\theta a}(b)$.
\end{theorem}

\begin{remark}\label{rmk:noPL}{\rm{
Observe that the right multiplication $r_{\theta \Omega'}$ map accounts for the extra term in (\ref{RBR}). Let us emphasize that in \cite{EM2} we formulated this theorem using the generic pre-Lie algebra structure underlying any associative Rota--Baxter algebra. Indeed, the above generalized adjoint operation, $\widetilde{ad}_{a}(b)=ad_{R(a)}(b) - r_{\theta a}(b)=[R(a),b]-\theta ba$,  defines a pre-Lie product on $\mathcal{A}$. However, we refrain from using the pre-Lie picture for reasons to become clear in the sequel.}}
\end{remark}

We introduce now the notion of the $R$-center in a Rota--Baxter algebra $\mathcal{A}$ of weight $\theta$, defined as the following subset of elements in $\mathcal{A}$:
$$
	Z_R(\mathcal{A}):=\{ x  \in\mathcal{A}\ |\ xR(a)=R(xa) \ {\rm{and}}\ R(a)x=R(ax),\;\ \forall a\in \mathcal{A}\}.  
$$
Observe that the $R$-center forms a subalgebra of $\mathcal{A}$. Indeed, for $u,v \in Z_R(\mathcal{A})$, we have $uvR(a)=uR(va)=R(uva)$, and $[R(a),x]=R([a,x])$. Notice the identity $R(x)=xR(\un)$, for any $x \in Z_R(\mathcal{A})$.  In the context of matrix valued differentiable function with the Riemann integral as weight zero Rota--Baxter map, the $R$-center includes amongst others the constant matrices. Another useful example are upper (or lower) triangular matrices, say, with unit diagonal, and entries in an arbitrary Rota--Baxter algebra $\mathcal{A}$ of weight $\theta$. They form a noncommutative Rota--Baxter algebra $M_\mathcal{A}$ of weight $\theta$ with the natural Rota--Baxter map. In this case the triangular matrices with entries strictly in $k$ are in the $R$-center.   

Now we would like to prove that the combinatorial structure underlying the  ${\mathsf{T}}^{\ast}$-operation, that is the solution of the IVP provided by Theorem \ref{thm:RB0}, extends to Rota--Baxter algebras. 

\begin{theorem} \label{thm:newMagnus} 
Let $(\mathcal{A},R)$ be a unital Rota--Baxter $k$-algebra of weight $\theta \in k$. Let $\alpha$ be an element in its $R$-center, and $\Omega'_\alpha:=\Omega'_\alpha(\lambda a)$, $a \in \mathcal{A}$, be the element of $\mathcal{A}[[\lambda]]$ such that:
$$
	X=\exp\bigl(\alpha + R(\Omega'_\alpha)\bigr),
$$ 
where $X$ is the solution of the linear fixpoint equation:
$$
	X = \exp(\alpha) + \lambda R(Xa).
$$ 
The element $\Omega'_\alpha$ obeys the following recursive equation:
\begin{equation}
\label{eq:twistMagnus}
    \Omega'_\alpha(\lambda a) = \frac{-ad_{\alpha + R(\Omega'_\alpha)} + r_{\theta \Omega'_\alpha}}
    	{e^{-ad_{\alpha + R(\Omega'_\alpha)} + r_{ \theta \Omega'_\alpha}}-1}(\lambda a)
            =\sum\limits_{m\ge 0} (-1)^m\frac{B_m}{m!}\ \widetilde{ad}^{(m)}_{\alpha,\Omega'_\alpha}(\lambda a),
\end{equation}
where $\widetilde{ad}_{\alpha,\Omega'_\alpha}(b):=ad_{\alpha+R(\Omega'_\alpha)}-r_{\theta \Omega'_\alpha}$.
\end{theorem}
Observe that this theorem is a generalization of Theorems \ref{thm:RB0} and \ref{thm:pLMagnus}. Returning to Remark \ref{rmk:noPL}, we see here that this generalization involves the operation $\widetilde{ad}_{a,b}:=ad_{a + R(b)} - r_{\theta b}$, which does not define a pre-Lie algebra product on the Rota--Baxter $k$-algebra $\mathcal{A}$.

\begin{proof}
We follow partly \cite{EM2}. First, note that we can decompose $ad_{\alpha + R(\Omega'_\alpha)}= \ell_{\alpha + R(\Omega'_\alpha)} - r_{\alpha + R(\Omega'_\alpha)}$, where $\ell_a(b):=ab$. It is clear that both operations commute, that is, $\ell_a r_b(c)=acb=r_b\ell_a (c)$. Now, let:
\allowdisplaybreaks{
\begin{eqnarray*}
	\lambda a &=&  \frac{e^{-ad_{\alpha + R(\Omega'_\alpha)}  + r_{\theta \Omega'_\alpha}}-1}{-ad_{\alpha 
				+ R(\Omega'_\alpha)}  + r_{\theta \Omega'_\alpha}}(\Omega'_\alpha)= \int_0^1 ds e^{-s(ad_{\alpha + R(\Omega'_\alpha)} - r_{\theta \Omega'_\alpha})}(\Omega'_\alpha)\\
	   &=& \int_0^1 ds e^{-s(\ell_{\alpha + R(\Omega'_\alpha)} - r_{\alpha + R(\Omega'_\alpha)}
	   					- r_{\theta \Omega'_\alpha})}(\Omega'_\alpha)= \int_0^1 ds e^{-s\ell_{\alpha + R(\Omega'_\alpha)}} e^{ s r_{\alpha + R(\Omega'_\alpha)}
	   					+ s r_{\theta  \Omega'_\alpha}}(\Omega'_\alpha).
\end{eqnarray*}}   
Now, we multiply by $X=\exp\bigl(\alpha + R(\Omega'_\alpha)\bigr)$ from the left:
\allowdisplaybreaks{
\begin{eqnarray*}
	 Xa  &=& \int_0^1 ds e^{(1-s)\ell_{\alpha + R(\Omega'_\alpha)}} e^{ s(r_{\alpha + R(\Omega'_\alpha)}
	   					+ r_{\theta  \Omega'_\alpha})}(\Omega'_\alpha)\\
	        &=& \int_0^1 ds \sum_{p,q\geq 0} \frac{(1-s)^ps^q}{q!p!}(\ell_{\alpha + R(\Omega'_\alpha)})^p
	  ( r_{\alpha + R(\Omega'_\alpha)}
	   					+ r_{\theta  \Omega'_\alpha})^q  (\Omega'_\alpha)\\
	        &=& \int_0^1 ds \sum_{p,q\geq 0} \frac{(1-s)^ps^q}{q!p!}(\ell_{\alpha + R(\Omega'_\alpha)})^p
	  ( r_{\alpha + R(\Omega'_\alpha) + \theta  \Omega'_\alpha})^q  (\Omega'_\alpha)\\
	        &=& \sum_{n > 0} \frac{1}{n!} \sum_{p+q = n-1}(\alpha + R(\Omega'_\alpha))^p
	         (\Omega'_\alpha) ( \alpha - \tilde{R}(\Omega'_\alpha))^q,	  			
\end{eqnarray*}}   
where we used that $-\tilde{R}=\theta \id + R$. Hence, we would like to show that $X-\exp(\alpha) = \lambda R(Xa)$. Expanding:
$$
	X- \exp(\alpha) = \sum_{n \geq 0} \frac{1}{n!}\bigl( (\alpha + R(\Omega'_\alpha))^n - \alpha^n\bigr), 
$$
lefts us with the goal to prove order by order the general identity:
$$
	\sum_{n > 0} \frac{1}{n!} \bigl((\alpha + R(\beta))^n - \alpha^n\bigr) =
	\sum_{n > 0} \frac{1}{n!} \sum_{p+q = n-1}R\bigl((\alpha + R(\beta))^p
	         (\beta) ( \alpha - \tilde{R}(\beta))^q\bigr),
$$
where $\alpha \in Z_R(\mathcal{A})$ and $\beta$ is an arbitary element in $\mathcal A$. 

The case $n=1$ is obviously true. Let us look at $n=2$. Then we have on the left hand side:
\allowdisplaybreaks{
\begin{eqnarray*}   
	(\alpha + R(\beta))^2 - \alpha^2 &=& \alpha R(\beta) 
				+R(\beta) \alpha + R(\beta)R(\beta)\\
	&=& R(\alpha\beta) +R(\beta \alpha)  + R(\beta)R(\beta).
\end{eqnarray*}}   
In the last line we used that $\alpha$ is in the $R$-center $Z_R(\mathcal{A})$. On the right hand side we find:
\allowdisplaybreaks{
\begin{eqnarray*}   
	R\bigl((\alpha + R(\beta))\beta +
	         \beta ( \alpha - \tilde{R}(\beta))\bigr) &=& R\bigl(\alpha \beta 
	         								+ R(\beta)\beta +
	         \beta \alpha + \beta R(\beta) + \theta \beta\beta\bigr)\\
	 &=& R(\alpha \beta + \beta \alpha) + R(\beta)R(\beta).
\end{eqnarray*}}   
In the last line we used the Rota--Baxter relation (\ref{RBR}). We show now the case $n+1$, assuming the identity holds up to order $n$. This yields:
\allowdisplaybreaks{
\begin{eqnarray*}   
	\lefteqn{(\alpha + R(\beta))^{n+1} - \alpha^{n+1} 
			= (\alpha + R(\beta))^{n}(\alpha + R(\beta)) - \alpha^{n}\alpha}\\
	&=&  (\alpha + R(\beta))^{n}R(\beta) + ((\alpha + R(\beta))^{n} - \alpha^{n})\alpha\\
	&=&  \sum_{p+q = n-1}R\bigl((\alpha + R(\beta))^p
	         			(\beta) ( \alpha - \tilde{R}(\beta))^q\bigr)R(\beta) 
					+ \alpha^{n}R(\beta) \\
	&&    \qquad +  \sum_{p+q = n-1}R\bigl((\alpha + R(\beta))^p
	         			(\beta) ( \alpha - \tilde{R}(\beta))^q\alpha\bigr)\\ 
	 &=&  \sum_{p+q = n-1}R\bigl((\alpha + R(\beta))^p
	         			(\beta) ( \alpha - \tilde{R}(\beta))^qR(\beta)\bigr)  \\
	&&     \qquad + \sum_{p+q = n-1}R\bigl(R((\alpha + R(\beta))^p
	         			(\beta) ( \alpha - \tilde{R}(\beta))^q\beta )\bigr)  \\
	&& 	 \qquad\quad + \sum_{p+q = n-1}R\bigl((\alpha + R(\beta))^p
	         			(\beta) ( \alpha - \tilde{R}(\beta))^q\theta \beta \bigr)  \\
	&&    \qquad\quad\; + R(\alpha^{n} \beta)+ \sum_{p+q = n-1}R\bigl((\alpha + R(\beta))^p
	         			(\beta) ( \alpha - \tilde{R}(\beta))^q\alpha\bigr)\\    
	&=&  \sum_{p+q = n-1}R\bigl((\alpha + R(\beta))^{p}
	         			(\beta) ( \alpha - \tilde{R}(\beta))^{q+1}\bigr) \\
	&& 	  \qquad + \sum_{p+q = n-1}R\bigl(R((\alpha + R(\beta))^p
	         			(\beta) ( \alpha - \tilde{R}(\beta))^q)\beta \bigr) 
					+ R(\alpha^{n}\beta )\\
	&=&  \sum_{p+q = n-1}R\bigl((\alpha + R(\beta))^{p}
	         			(\beta) ( \alpha - \tilde{R}(\beta))^{q+1}\bigr) \\
	&&	 \qquad  + R\bigl((\alpha + R(\beta))^{n}\beta - \alpha^{n}\beta\bigr) 
	          		+ R(\alpha^{n}\beta )\\ 
	 &=&  \sum_{p+q = n-1}R\bigl((\alpha + R(\beta))^{p}
	         (\beta) ( \alpha - \tilde{R}(\beta))^{q+1}\bigr) 
	         + R\bigl((\alpha + R(\beta))^{n} \beta\bigr) \\
	  &=&  \sum_{p+q = n}R\bigl((\alpha + R(\beta))^{p}
	         (\beta) ( \alpha - \tilde{R}(\beta))^{q}\bigr)                               
\end{eqnarray*}}   
where we recall that $\alpha - \tilde{R}(\beta)=\alpha + \theta \beta + R(\beta)$.   
\end{proof}

Note that in the context of the weight zero RB algebra defined in terms of the Riemann integral, the generalization following from the theorem allows to absorb the initial value in (\ref{eq:IVP}) directly into the Magnus' expansion. 

\begin{remark}\label{rmk:BaAt}
{\rm{Recall Remark \ref{rmk:BaAtVo} and the fact that any Rota--Baxter algebra has two Rota--Baxter operators, $R$ and $\tilde{R}$. In \cite{Atkinson,Baxter} Baxter and Atkinson showed that the solution to the recursions:
$$
	E=a+\lambda R(bE) \quad {\rm{and}} \quad F=R(a)+\lambda R(bF)
$$
are given by $F=R(aX)Y$ and $E= a + R(abX)Y$, respectively, where $X,Y$ are solutions to:
$$
	X=\un + \lambda R(Xb) \quad {\rm{resp.}} \quad Y=\un + \lambda \tilde{R}(bY).
$$}}
\end{remark}

Together with the solutions $X$ and $Y$ for the last two fixpoint equations comes the so-called Atkinson factorization, which generalizes straightforwardly in the context of the above theorem. 

\begin{proposition}\label{prop:Atkin}
Let $(\mathcal{A},R)$ be a unital Rota--Baxter $k$-algebra of weight $\theta \in k$. Let $\alpha,\beta$ be in its $R$-center $Z_R(\mathcal{A})$. Then we have a factorization:
$$
\un- \lambda a = X^{-1}\exp(\gamma)Y^{-1},
$$ 
where:
$$
	X=\exp(\alpha) + \lambda R(Xa) \qquad\qquad Y=\exp(\beta) + \lambda \tilde{R}(aY), 
$$  
and $\exp(\gamma)=\exp(\alpha) \exp(\beta)$, $\gamma \in Z_R(\mathcal{A})$.
\end{proposition}

\begin{proof}
The proof reduces to a simple verification. 
\allowdisplaybreaks{
\begin{eqnarray*}
	XY &=& \exp(\gamma) + \lambda R(Xa) \exp(\beta) 
			+  \lambda\exp(\alpha)\tilde{R}(aY)+  \lambda^2R(Xa)\tilde{R}(aY)\\
	      &=& \exp(\gamma)  +  \lambda R(Xa\exp(\beta))  
	      		+ \lambda \tilde{R}(\exp(\alpha)aY) +  \lambda^2R(Xa)\tilde{R}(aY)\\
	      &=&  \exp(\gamma)  + \lambda R\bigl(Xa(\exp(\beta) +  \lambda\tilde{R}(aY))\bigr)  
	      		+ \lambda \tilde{R}\bigl((\exp(\alpha) +  \lambda R(Xa))aY\bigr)\\
	      &=& \exp(\gamma) + \lambda XaY,				     
\end{eqnarray*}}  
from which the result follows. 
\end{proof}


\section{Example: Finite differences}
\label{ssect:finitediff}

As an example of an interesting fixpoint equation involving a non-zero weight Rota--Baxter structure, we will now approach the particular case of a linear finite difference initial value problem. For this, we define the finite time-difference operator $\Delta$ such that for all functions $f$ defined on $\mathbb{N}$:
\begin{equation*}
	\Delta(f)(n):=f(n+1) - f(n).
\end{equation*}
Observe that this map satisfies a generalized Leibniz rule:
\begin{eqnarray*}
\Delta(fg)(n) &=& (fg)(n+1)-(fg)(n)=f(n+1)g(n+1)-f(n)g(n)\\
		   &=& f(n)\Delta(g)(n)+\Delta(f)(n)g(n) + \Delta(f)(n)\Delta(g)(n)
\end{eqnarray*}
Recall the finite Riemann sum map (\ref{Riemsum}), $R(X)_n:=\sum\limits_{k=0}^{n-1}X_k$. The latter satisfies the Rota--Baxter relation of weight $\theta$. We put the weight equal to one, $R:=R_1$. Then:
\begin{eqnarray*}
	R(\Delta(f))(m)=\sum_{n=0}^{m-1}\Delta(f)(n) &=&\sum_{n=1}^{m-1} (f(n+1) - f(n))\\
			   &=&\sum_{n=1}^{m} f(n) - \sum_{n=0}^{m-1} f(n)\\
			   &=& f(m)-f(0).
\end{eqnarray*}
We are interested in the finite difference initial value problem: 
\begin{equation}
	\begin{array}{l}
	\Delta X_{n}=X_{n}B_{n}\\
	X_{0}=\exp\left(A_{0}\right) 
	\end{array}
\label{eq:findif}
\end{equation}
where $X_n:=X(n)$ and $B_n\not= -1$. Applying $R$ leads to a generalized Atkinson equation: 
\begin{equation*}
	X=\exp\left(A_{0}\right)+R(XB),
\end{equation*}
which yields:
\begin{equation*}
	X_{n}=\exp\left(A_{0}\right)+\sum_{k=0}^{n-1}X_{k}B_{k}.
\end{equation*}
Observe that the solution can be written as a finite product, as well as, using Theorem \ref{thm:pLMagnus}, as an exponential: 
$$
	X_{n}=\exp\left(A_{0}\right)\overrightarrow{\prod_{k=0}^{n-1}}\left(1+B_{k}\right)=\exp\left(R(\Omega'(n))\right).
$$

Now, by applying Theorem \ref{thm:newMagnus}, with $\alpha =A(0)$ and $\Omega'_\alpha=\Delta A$ (recall that $X=\exp(A)$ and $R(\Delta A)=A-A(0)$),
we get:

\begin{theorem}
 With our previous notation, for the finite difference IVP (\ref{eq:findif}) we get:
 \begin{equation*}
	B_{k}=\left(\frac{\exp\left(-\ell_{A_{k}}+r_{A_{k+1}}\right)-Id}
	                       {-\ell_{A_{k}}+r_{A_{k+1}}}\right)\left[\Delta A_{k}\right].
\end{equation*}
where as before $\ell_{X}$ is the left multiplication by $X$ and $r_{X}$ is the right multiplication by $X$.
\end{theorem}

In particular:
\begin{equation}
	\exp\left(-A_{0}\right)\exp\left(A_{N}\right)
	=\overrightarrow{\prod_{k=0}^{N-1}}\left(1+\left(\frac{\exp\left(-\ell_{A_{k}}+r_{A_{k+1}}\right)-1}
	{-\ell_{A_{k}}+r_{A_{k+1}}}\right)\left[\Delta A_{k}\right]\right)\label{eq:dismag}
\end{equation}
 
Note that these formulas can also be obtained from a straightforward calculation. We include a direct proof for the sake of transparency. The starting point is a discrete analog of Duhamel's formula (\ref{eq:duh}):

\begin{lemma}
We have:
\begin{equation}
	\exp\left(-A_{k}\right)\Delta\left[\exp\left(A_{k}\right)\right]
	=\int_{0}^{1}d\nu \exp\left(-\nu A_{k}\right)\left(\Delta A_{k}\right)\exp\left(\nu A_{k+1}\right).
\label{eq:disduh}
\end{equation}
\end{lemma}

The proof follows by taking the integral from $0$ to $1$ of the following equality:
\begin{eqnarray*}
	\frac{d}{d\nu}\left[\exp\left(-\nu A_{k}\right)\Delta\left[\exp\left(\nu A_{k}\right)\right]\right] 
	& = & -\exp\left(-\nu A_{k}\right)A_{k}\Delta\left[\exp\left(\nu A_{k}\right)\right]
		+\exp\left(-\nu A_{k}\right)\Delta\left[A_{k}\exp\left(\nu A_{k}\right)\right]\\
 	& = & -\exp\left(-\nu A_{k}\right)A_{k}\Delta\left[\exp\left(\nu A_{k}\right)\right]
		+\exp\left(-\nu A_{k}\right)A_{k}\Delta\left[\exp\left(\nu A_{k}\right)\right]\\
 	&    & +\exp\left(-\nu A_{k}\right)\Delta\left[A_{k}\right]\exp\left(\nu A_{k}\right)
		+\exp\left(-\nu A_{k}\right)\Delta\left[A_{k}\right]\Delta\left[\exp\left(\nu A_{k}\right)\right]\\
	& = & \exp\left(-\nu A_{k}\right)\Delta\left[A_{k}\right]\exp\left(\nu A_{k+1}\right).
\end{eqnarray*}

Now we verify identity (\ref{eq:dismag}). The first thing to remark is that:
\begin{equation*}
	\exp\left(-A_{0}\right)\exp\left(A_{N}\right)
	=\overrightarrow{\prod_{k=0}^{N-1}}\left(1+\exp\left(-A_{k}\right)\Delta\left[\exp\left(A_{k}\right)\right]\right)\label{eq:dismag-1}
\end{equation*}
follows when we note that $1+\exp\left(-A_{k}\right)\Delta\left[\exp\left(A_{k}\right)\right]=\exp\left(-A_{k}\right)\exp\left(A_{k+1}\right).$ Then, we first apply Duhamel's formula (\ref{eq:disduh}), and second expand $\exp\left(-\nu A_{k}\right)$ and $\exp\left(\nu A_{k+1}\right),$ and eventually we do the explicit integration of $\nu$. This yields:
 \begin{equation}
	\exp\left(-A_{0}\right)\exp\left(A_{N}\right)
	=\overrightarrow{\prod_{k=0}^{N-1}}\left(1+\sum_{N=1}^{\infty}\frac{1}{N!}
	\sum_{n=0}^{N-1}\left(-1\right)^{n}C_{N-1}^{n}\left(A_{k}\right)^{n}
	\Delta\left[A_{k}\right]\left(A_{k+1}\right)^{N-1-n}\right).
\label{eq:dismag-1-1}
\end{equation}

Finally, we note that: 
 \begin{equation*}
	\sum_{n=0}^{N-1}\left(-1\right)^{N-1-n}C_{N-1}^{n}\left(A_{k}\right)^{n}
	\Delta\left[A_{k}\right]\left(A_{k+1}\right)^{N-1-n}
	=\left(l_{A_{k}}-r_{A_{k+1}}\right)^{N-1}\Delta\left[A_{k}\right],
\end{equation*}
such that:
 \begin{eqnarray*}
	\sum_{N=1}^{\infty}\frac{1}{N!}\sum_{n=0}^{N-1}\left(-1\right)^{n}C_{N-1}^{n}\left(A_{k}\right)^{n}
	\Delta\left[A_{k}\right]\left(A_{k+1}\right)^{N-1-n} 
	& = & \sum_{N=1}^{\infty}\frac{\left(-1\right)^{N-1}}{N!}\left(l_{A_{k}}-r_{A_{k+1}}\right)^{N-1}
		\Delta\left[A_{k}\right]\\
 	& = & \left(\frac{\exp\left(-l_{A_{k}}+r_{A_{k+1}}\right)-Id}{-l_{A_{k}}+r_{A_{k+1}}}\right)
		\Delta\left[ A_{k}\right].
\end{eqnarray*}
This relation, together with (\ref{eq:dismag-1-1}) concludes the proof of (\ref{eq:dismag}).\\

\bigskip

{\bf{Acknowledgements}} The third author is supported by a Ram\'on y Cajal research grant from the Spanish government. We thank D.~Manchon for helpful discussions. We thank the CNRS (GDR Renormalisation) for support.\\


\end{document}